\documentclass[11pt,a4paper]{article}

\usepackage[utf8]{inputenc}
\usepackage[T1]{fontenc}
\usepackage{lmodern}
\usepackage{microtype}

\usepackage{amsmath,amssymb,amsthm}
\usepackage{mathtools}
\usepackage{thmtools}
\usepackage{bm}

\usepackage{hyperref}
\usepackage{cleveref}
\usepackage{csquotes}

\usepackage[margin=1in]{geometry}
\usepackage{enumitem}
\usepackage{xcolor}

\usepackage{algorithm}
\usepackage{algpseudocode}

\hypersetup{
  colorlinks=true,
  linkcolor=blue,
  citecolor=blue,
  urlcolor=blue,
  pdftitle={Algorithmic Barriers to Detecting and Repairing Structural Overspecification in Adaptive Data-Structure Selection},
  pdfauthor={Faruk Alpay; Levent Sarioglu},
  pdfkeywords={data structures, graph algorithms, string algorithms, dynamic data structures, algorithm selection, decidability, fixed point, program transformation, Bradley-Terry-Luce}
}

\newtheorem{theorem}{Theorem}[section]
\newtheorem{lemma}[theorem]{Lemma}
\newtheorem{corollary}[theorem]{Corollary}
\newtheorem{proposition}[theorem]{Proposition}

\theoremstyle{definition}
\newtheorem{definition}[theorem]{Definition}

\theoremstyle{remark}
\newtheorem{remark}[theorem]{Remark}

\newcommand{\Sig}{\Sigma}
\newcommand{\Str}{\Sig^{\ast}}

\newcommand{\Fclass}{\mathcal{F}}

\newcommand{\Rreward}{R}
\newcommand{\Smodel}{S}
\newcommand{\Warrant}{W}
\newcommand{\Vrel}{V}
\newcommand{\valscore}{v}

\newcommand{\N}{\mathbb{N}}

\newcommand{\R}{\mathbb{R}}

\newcommand{\KL}{\mathrm{KL}}

\DeclareMathOperator{\Prb}{Pr}

\title{
  Algorithmic Barriers to Detecting and Repairing\\
  Structural Overspecification in Adaptive Data-Structure Selection
}

\author{
  Faruk Alpay\thanks{
    Department of Computer Engineering, Bahcesehir University.
    \texttt{faruk.alpay@bahcesehir.edu.tr}
  }
  \and
  Levent Sarioglu\thanks{
    Department of Computer Engineering, Bahcesehir University.
    \texttt{levent.sarioglu@bahcesehir.edu.tr}
  }
}

\date{}

\begin{document}
\maketitle

\begin{abstract}
We study algorithmic barriers to detecting and repairing a systematic form of structural overspecification in adaptive data-structure selection.
An input instance induces an implied workload signature---ordering, sparsity, dynamism, locality, or substring structure---and candidate implementations may be preferred because they match that full signature even when the measured workload evidence supports only a strict subset of it.
Working under a model in which pairwise evaluators favor implementations that realize the implied signature, we show that this preference propagates through both benchmark aggregation and Bradley--Terry--Luce fitting.
We then establish two results.
First, a sharp decidability boundary: determining whether a representation-selection pipeline exhibits structural commitment beyond measured warrant is undecidable on unbounded input domains (by reduction from the halting problem) but decidable with exponential enumeration cost on finite domains.
Second, a fixed-point barrier: under a conservative program-transformation constraint---the repair operator must leave already evidence-aligned pipelines unchanged---any total computable repair operator admits an overspecified fixed point, constructed via Kleene's recursion theorem.
These barriers are qualitatively distinct from classical lower bounds in data-structure design: cell-probe and dynamic-graph lower bounds constrain the efficiency of answering operations on finite workloads, while our results constrain the possibility of uniformly detecting and repairing overspecification across pipeline families.
\end{abstract}


\medskip

\section{Introduction}
\label{sec:intro}

Selecting an implementation or representation for a workload---choosing between adjacency lists and adjacency matrices, balanced search trees and hashing, suffix arrays and suffix trees, or static and dynamic graph primitives---is a central task in data structures and algorithm design~\cite{cormen2009,sleator1983,holm2001,ukkonen1995}.
Modern systems increasingly make such choices from traces, benchmark outcomes, or learned cost models, often using pairwise comparisons between candidate implementations and then aggregating those outcomes into a single score or ranking~\cite{rice1976,bradley1952,mcgeoch2012,kraska2018}.

The complexity of data-structure design is well studied:
dynamic data structures admit update/query trade-offs~\cite{fredman1989},
dynamic graph algorithms have finely characterized complexity frontiers~\cite{holm2001,demetrescu2004},
and string indexing structures achieve different preprocessing/query regimes~\cite{kmp1977,manber1993,ukkonen1995}.
These results are \emph{complexity barriers}: they bound the efficiency of supporting a given set of operations on a fixed finite instance or operation sequence.

This paper identifies a different type of barrier---a \emph{computability barrier}---that arises when the representation-selection process is systematically overspecified.
Specifically, we formalize a structural-preference channel:
when an input instance suggests a workload signature, evaluators may prefer implementations that realize the full implied signature even when the measured evidence only warrants a smaller subset.
For example, a sparse graph workload may be mapped to aggressively dynamic graph machinery without evidence of adversarial updates; a string-processing task may trigger heavy suffix-based indexing from weak locality cues; a dictionary workload may be assigned ordered search structures from incidental sortedness hints.
Under standard pairwise aggregation models, this preference provably propagates into the learned scoring function (Propositions~\ref{prop:social_inheritance_profile} and~\ref{prop:btl_inheritance}).
The algorithmic questions are then: can one \emph{detect} such overspecification, and can one \emph{repair} it uniformly?

\paragraph{Our contributions.}
We establish two main results, each with explicit scope audits.
\begin{enumerate}[label=(\roman*),leftmargin=*]
\item \textbf{Decidability boundary (Section~\ref{sec:computability}).}
The problem of deciding whether a representation-selection pipeline exhibits structural commitment beyond measured warrant is undecidable on unbounded input domains (Theorem~\ref{thm:rice_fixed}), via reduction from the halting problem.
On finite domains, the same predicate is decidable by exhaustive enumeration at exponential cost (Proposition~\ref{prop:finite_domain_decidable}).
This is a sharp computability boundary.

\item \textbf{Fixed-point barrier (Section~\ref{sec:correction}).}
Under a conservative constraint on program transformers---the repair operator must leave evidence-aligned pipelines unchanged---any total computable repair operator admits an overspecified fixed point (Theorem~\ref{thm:biased_fixed_point}), constructed via Kleene's recursion theorem.
No conservative repair operator can uniformly eliminate overspecification across all pipelines.
\end{enumerate}

These barriers are qualitatively distinct from classical data-structure lower bounds:
cell-probe and dynamic-graph lower bounds constrain time and space trade-offs on finite workloads;
our results are computability barriers that no amount of computation can overcome on unbounded pipeline families.

\paragraph{Paper structure.}
Section~\ref{sec:framework} introduces the formal framework.
Sections~\ref{sec:bias_model_inheritance} and~\ref{sec:bias_model_asymmetry} specify the overspecification model:
they establish that unwarranted structural preference propagates through benchmark aggregation and does so asymmetrically.
These are preconditions, not contributions---they formalize the problem that our main results address.
Sections~\ref{sec:computability} and~\ref{sec:correction} contain the main algorithmic contributions.

\subsection{Related work}
\label{sec:related}

\paragraph{Adaptive data structures and algorithm selection.}
Rice~\cite{rice1976} formulated the algorithm selection problem, emphasizing that the choice of algorithm depends on instance features and performance models.
Workload-sensitive and self-adjusting structures, such as dynamic trees~\cite{sleator1983} and splay trees~\cite{sleator1985}, show that representation matters structurally.
Learned and benchmark-driven data systems make this dependence even more explicit~\cite{mcgeoch2012,kraska2018}.
Our paper identifies a different barrier: not efficient adaptation, but \emph{undecidability} of detecting unwarranted structural commitment and \emph{incompleteness} of conservative repair.

\paragraph{Graph algorithms and dynamic data structures.}
Dynamic connectivity, shortest paths, and dynamic tree representations exhibit rich update/query trade-offs~\cite{sleator1983,holm2001,demetrescu2004}.
This literature studies efficient support for changing workloads, whereas our focus is meta-algorithmic: whether a selection pipeline can be uniformly audited for structural overspecification.

\paragraph{String algorithms and indexing structures.}
Classical string-processing results such as Knuth--Morris--Pratt~\cite{kmp1977}, suffix arrays~\cite{manber1993}, and suffix trees~\cite{ukkonen1995} demonstrate multiple competing representation regimes for the same underlying task.
Our framework treats such alternatives as candidates in a representation-selection pipeline and asks when a selector commits to too much structure relative to evidence.

\paragraph{Program repair and computability.}
Consistent query answering~\cite{arenas1999,bertossi2011} studies querying inconsistent data under minimal repairs.
Rice's theorem~\cite{rice1953} and fixed-point constructions via the recursion theorem underpin the computability barriers used here.
Our conservative repair operator (Definition~\ref{def:conservative}) is the program-transformation analogue of minimal repair.

\section{Formal framework}
\label{sec:framework}

\subsection{Instances, implementations, and computable selection pipelines}

Fix a finite alphabet $\Sig$.
Let $\Str=\Sig^\ast$ denote the set of all finite strings over $\Sig$.
We model workload instances as strings $x\in\Str$ and candidate implementations or representations as strings $y\in\Str$.

\begin{definition}[Representation-selection pipeline class]
\label{def:gen_class}
A \emph{representation-selection pipeline class} $\Fclass$ is a recursively enumerable set of total computable functions
$f:\Str\to\Str$.
Equivalently, there exists an effective enumeration $\{f_e\}_{e\in\N}$
such that $\Fclass=\{f_e:e\in\N\}$ and each $f_e$ halts on every input.
\end{definition}

\begin{remark}
\label{rem:finite_domain_decidable}
Restricting the domain to $\Sig^{\le n}$ makes it finite and many existential predicates decidable.
This is critical for the decidability boundary in Section~\ref{sec:computability}.
\end{remark}

\subsection{Workload signatures, measured warrant, and structural compatibility}

\begin{definition}[Workload-signature extractor]
\label{def:self_model}
A \emph{workload-signature extractor} is a total computable function
$\Smodel:\Str\to 2^{\Str}$
mapping each instance $x$ to a finite set $\Smodel(x)$ of structural feature strings.
These features may encode properties such as sorted access, sparsity, online updates, locality, or substring reuse.
\end{definition}

\begin{definition}[Measured-warrant extractor]
\label{def:warrant}
A \emph{measured-warrant extractor} is a total computable function
$\Warrant:\Str\to 2^{\Str}$
such that $\Warrant(x)\subseteq \Smodel(x)$ for all $x$,
with $\Warrant(x)$ interpreted as the structural features actually supported by the observed trace, benchmark, or specification.
\end{definition}

\begin{definition}[Structural compatibility relation and scores]
\label{def:validation_relation}
A \emph{structural compatibility relation} is a total computable function
$\Vrel:\Str\times\Str\to\{-1,0,+1\}$.
The \emph{structural compatibility score} is
$\valscore(x,y)=\sum_{s\in\Smodel(x)} \Vrel(y,s)$.
The \emph{beyond-warrant overspecification score} is
\begin{equation}
\label{eq:beyond_warrant}
\valscore_{\mathrm{bw}}(x,y)=\sum_{s\in \Smodel(x)\setminus \Warrant(x)} \Vrel(y,s).
\end{equation}
A positive value indicates that implementation $y$ instantiates structure suggested by the full signature of $x$ but not warranted by measured evidence.
\end{definition}

\begin{remark}
\label{rem:two_notions}
Any property of the form ``property of the input/output behavior of a total computable function''
falls under Rice-style limits on unbounded domains (Section~\ref{sec:computability}).
\end{remark}

\subsection{Benchmark aggregation and pairwise score fitting}

\begin{definition}[Benchmark profile and implementation aggregation]
\label{def:social_aggregation}
Fix $k$ evaluators, an instance $x$, and candidate set $Y_x$.
A \emph{benchmark profile} is $\pi=(\succ_1,\dots,\succ_k)$, where each $\succ_i$ ranks candidate implementations.
An \emph{implementation aggregation rule} $A$ maps such profiles to aggregate orders.
\end{definition}

\begin{definition}[Decisive benchmark family]
\label{def:decisive_pair}
$C\subseteq[k]$ is \emph{decisive for $(y,y')$}
if $(\forall i\in C:y\succ_i y') \Rightarrow y\succ_A y'$,
regardless of the preferences of the remaining evaluators.
\end{definition}

\begin{definition}[Logistic pairwise-implementation likelihood]
\label{def:logistic_likelihood}
Let $\sigma(t)=\frac{1}{1+e^{-t}}$.
A scoring function $\Rreward:\Str\times\Str\to\R$ induces
$\Prb[y\succ y'\mid x]=\sigma(\Rreward(x,y)-\Rreward(x,y'))$.
Bradley--Terry--Luce fitting provides a standard pairwise scoring model for such outcomes~\cite{bradley1952,hunter2004}.
\end{definition}

\section{Overspecification model: Inheritance through benchmark aggregation}
\label{sec:bias_model_inheritance}

This section specifies the structural overspecification model under which our algorithmic results operate.
The propositions here are not our main contributions; they formalize the precondition that unwarranted structural preference \emph{does} propagate through standard benchmark aggregation, establishing that detection and repair are problems worth studying.

\subsection{Deterministic inheritance under decisiveness}

\begin{definition}[Signature-monotone evaluator]
\label{def:validation_monotone_pair}
An evaluator ranking $\succ_i$ is \emph{signature-monotone on $(y,y')$} if
$\valscore(x,y)>\valscore(x,y') \Rightarrow y\succ_i y'$.
\end{definition}

\begin{proposition}[Inheritance under decisive benchmark families]
\label{prop:social_inheritance_profile}
Fix $x$, $Y_x$ with $|Y_x|\ge 3$, and an implementation aggregation rule $A$.
If there exists a family $C$ decisive for $(y,y')$
with every $i\in C$ signature-monotone on $(y,y')$,
then for every benchmark profile $\pi$:
$\valscore(x,y)>\valscore(x,y') \Rightarrow y\succ_A y'$.
\end{proposition}

\begin{proof}
Signature-monotonicity forces each $i\in C$ to rank $y\succ_i y'$; decisiveness forces $y\succ_A y'$.
\end{proof}

\subsection{Statistical inheritance under pairwise score fitting}

\begin{definition}[Random-utility evaluators]
\label{def:random_utility}
For evaluator $i$, latent utility is
$u_i(x,y)=q_i(x,y)+\alpha_i\,\valscore(x,y)$,
where $q_i$ is structural-preference-independent and $\alpha_i\in\R$ is signature sensitivity.
Pairwise outcomes follow $\Prb_i[y\succ y'\mid x]=\sigma(u_i(x,y)-u_i(x,y'))$.
\end{definition}

\begin{proposition}[Pairwise-fit inheritance under quality control]
\label{prop:btl_inheritance}
Assume Definition~\ref{def:random_utility} with $q_i(x,y)=q_i(x,y')$ for all $i$,
$\alpha_i\ge 0$, and $\bar{\alpha}>0$.
Let $\Delta v=\valscore(x,y)-\valscore(x,y')$.
If $\Delta v>0$, then
$\Prb[y\succ y'\mid x]>\tfrac{1}{2}$
and the probability is strictly increasing in $\Delta v$.
\end{proposition}

\begin{proof}
Under quality control,
$\Prb[y\succ y'\mid x]=\frac{1}{k}\sum_i\sigma(\alpha_i\Delta v)$.
Each term $\ge\tfrac12$; at least one is strict. Averaging gives $>\tfrac12$.
Strict monotonicity:
$\frac{d}{d\Delta v}\sigma(\alpha_i\Delta v)=\alpha_i\sigma(\alpha_i\Delta v)(1-\sigma(\alpha_i\Delta v))\ge 0$,
strict for $\alpha_{i_0}>0$.
\end{proof}

\begin{corollary}[Learned scores inherit structural overspecification]
\label{cor:consistency_to_reward}
Under Proposition~\ref{prop:btl_inheritance}'s conditions,
any consistent pairwise score estimator satisfies
$\Delta v>0 \Rightarrow \widehat{R}(x,y)>\widehat{R}(x,y')$
in the large-sample limit.
\end{corollary}

\begin{proof}
Consistency gives $\widehat p\to p^\star>\tfrac12$; the logistic link converts this to a positive score difference.
\end{proof}

\begin{remark}
\label{rem:where_fails_btl}
Proposition~\ref{prop:btl_inheritance} can fail under mixed-sign $\alpha_i$, mismeasured workload traces, insufficient pair coverage, or explicit regularizers penalizing structural complexity.
\end{remark}

\section{Overspecification model: Asymmetric underprovision penalties}
\label{sec:bias_model_asymmetry}

The model is completed by noting that missing structure and adding extra structure need not be penalized symmetrically.
In algorithm engineering, underprovision is often treated as more costly than overprovision:
failing to provide dynamic updates or locality support can be penalized more heavily than carrying redundant machinery.
This asymmetry helps explain why repair is harder than detection:
a repair operator must undo a direction-dependent distortion, not merely subtract a constant.

\begin{definition}[Piecewise-linear regret weighting]
\label{def:prospect_weight}
Evaluator $i$ uses $w_i(\delta)=\alpha_i\delta$ for $\delta\ge 0$ and $w_i(\delta)=-\lambda_i\alpha_i|\delta|$ for $\delta<0$,
where $\alpha_i>0$ and $\lambda_i\ge 1$~\cite{tversky1992,kahneman1979}.
With $u_i(x,y)=q_i(x,y)+w_i(\valscore(x,y))$,
define the population score $\Rreward(x,y)=\frac{1}{k}\sum_i u_i(x,y)$.
\end{definition}

\begin{proposition}[Sensitivity-weighted underprovision ratio]
\label{prop:asymmetry_weighted}
Let $y^+,y^0,y^-$ have structural compatibility scores $+\delta,0,-\delta$ with equal baseline quality.
Define $\lambda_{\mathrm{eff}}:=\sum_i\lambda_i\alpha_i/\sum_i\alpha_i$.
Then
$\Rreward(x,y^0)-\Rreward(x,y^-)
=\lambda_{\mathrm{eff}}\cdot(\Rreward(x,y^+)-\Rreward(x,y^0))$,
with $\min_i\lambda_i\le\lambda_{\mathrm{eff}}\le\max_i\lambda_i$.
\end{proposition}

\begin{proof}
Under equal baseline quality:
$\Rreward(x,y^+)-\Rreward(x,y^0)=\frac{1}{k}\sum_i\alpha_i\delta$
and
$\Rreward(x,y^0)-\Rreward(x,y^-)=\frac{1}{k}\sum_i\lambda_i\alpha_i\delta$.
Dividing gives $\lambda_{\mathrm{eff}}$, a convex combination of $\lambda_i$.
\end{proof}

\section{Main result I: Decidability boundary for overspecification detection}
\label{sec:computability}

We now present our first main result: a sharp decidability boundary for detecting structural overspecification in representation-selection pipelines.

\subsection{The overspecification predicate}

\begin{definition}[Beyond-warrant overspecification predicate]
\label{def:bias_predicate_bw}
For a total computable pipeline $f:\Str\to\Str$,
\begin{equation}
\label{eq:bias_pred}
B_{\mathrm{bw}}(f)=
\begin{cases}
1, & \exists x\in\Str:\;\valscore_{\mathrm{bw}}(x,f(x))>0,\\
0, & \text{otherwise.}
\end{cases}
\end{equation}
\end{definition}

\begin{definition}[Index-set form]
\label{def:bias_index_set}
Fix an effective enumeration $\{f_e\}_{e\in\N}$ of total computable pipelines.
The corresponding index set is
\[
\mathcal{B}_{\mathrm{bw}}:=\{e\in\N: B_{\mathrm{bw}}(f_e)=1\}.
\]
\end{definition}

\begin{lemma}[Extensionality]
\label{lem:bw_extensional}
If $f,g:\Str\to\Str$ are total computable and $f(x)=g(x)$ for all $x\in\Str$,
then $B_{\mathrm{bw}}(f)=B_{\mathrm{bw}}(g)$.
Equivalently, if $f_e=f_{e'}$ extensionally, then
$e\in\mathcal{B}_{\mathrm{bw}} \Leftrightarrow e'\in\mathcal{B}_{\mathrm{bw}}$.
\end{lemma}

\begin{proof}
By pointwise equality, for every $x$ we have
$\valscore_{\mathrm{bw}}(x,f(x))=\valscore_{\mathrm{bw}}(x,g(x))$.
Hence the existential condition in~\eqref{eq:bias_pred} is identical for $f$ and $g$.
\end{proof}

\begin{lemma}[Non-triviality]
\label{lem:nontrivial_bw}
Assume there exist $x_0\in\Str$, $s_0\in\Str$, $y^+\in\Str$, and a padding symbol $\#\in\Sig$
such that for every $n\in\N$,
\[
s_0\in\Smodel(x_0\#^n)\setminus\Warrant(x_0\#^n)
\quad\text{and}\quad
\Vrel(y^+,s_0)=+1.
\]
Then $B_{\mathrm{bw}}$ is non-trivial:
$f_1\equiv y^+$ gives $B_{\mathrm{bw}}(f_1)=1$;
$f_2\equiv\epsilon$ (with $\Vrel(\epsilon,s)=0$ for all $s$) gives $B_{\mathrm{bw}}(f_2)=0$.
\end{lemma}

\begin{proof}
For $f_1\equiv y^+$, evaluate at $x_0$:
\[
\valscore_{\mathrm{bw}}(x_0,f_1(x_0))
=\sum_{s\in\Smodel(x_0)\setminus\Warrant(x_0)}\Vrel(y^+,s)
\ge \Vrel(y^+,s_0)=1>0,
\]
so $B_{\mathrm{bw}}(f_1)=1$.
For $f_2\equiv\epsilon$, by assumption $\Vrel(\epsilon,s)=0$ for all $s$, hence
$\valscore_{\mathrm{bw}}(x,f_2(x))=0$ for every $x$, so $B_{\mathrm{bw}}(f_2)=0$.
\end{proof}

\begin{remark}
\label{rem:padding_condition}
The padding condition in Lemma~\ref{lem:nontrivial_bw} is a technical device
for the halting reduction: it provides an unbounded \enquote{clock} family
$x_0\#^n$ while preserving a fixed beyond-warrant witness $s_0$.
It holds whenever $\Smodel$ and $\Warrant$ are invariant under this padding.
\end{remark}

\subsection{Undecidability on unbounded domains}

\begin{theorem}[Undecidability of overspecification detection]
\label{thm:rice_fixed}
Under the conditions of Lemma~\ref{lem:nontrivial_bw},
there is no Turing machine that decides $B_{\mathrm{bw}}(f)$ for all total computable $f$.
\end{theorem}

\begin{proof}
Let $H=\{\langle M,w\rangle: M\text{ halts on input }w\}$.
Fix witnesses $x_0,s_0,y^+,\epsilon,\#$ from Lemma~\ref{lem:nontrivial_bw}.

For each pair $\langle M,w\rangle$, define a total computable function $f_{M,w}:\Str\to\Str$ by
\[
f_{M,w}(x)=
\begin{cases}
y^+, & x=x_0\#^n\text{ for some }n\in\N\text{ and }M(w)\text{ halts within }n^2\text{ steps},\\
\epsilon, & \text{otherwise.}
\end{cases}
\]
Totality is immediate because bounded simulation to $n^2$ steps always halts.
By the $s$-$m$-$n$ theorem, there is a total computable map
$r:\langle M,w\rangle\mapsto e_{M,w}$ such that $f_{e_{M,w}}=f_{M,w}$.

We now verify correctness of the reduction.

\emph{If $\langle M,w\rangle\in H$:}
suppose $M(w)$ halts in exactly $t$ steps.
Choose $n\ge\lceil\sqrt t\rceil$; then $n^2\ge t$, hence
$f_{M,w}(x_0\#^n)=y^+$.
Therefore
\[
\valscore_{\mathrm{bw}}(x_0\#^n,f_{M,w}(x_0\#^n))
=\sum_{s\in\Smodel(x_0\#^n)\setminus\Warrant(x_0\#^n)}\Vrel(y^+,s)
\ge\Vrel(y^+,s_0)=1>0,
\]
so $B_{\mathrm{bw}}(f_{M,w})=1$.

\emph{If $\langle M,w\rangle\notin H$:}
no bounded simulation ever observes halting, so $f_{M,w}(x)=\epsilon$ for all $x$.
Hence $\valscore_{\mathrm{bw}}(x,f_{M,w}(x))=0$ for all $x$, and $B_{\mathrm{bw}}(f_{M,w})=0$.

Thus,
\[
\langle M,w\rangle\in H
\iff B_{\mathrm{bw}}(f_{M,w})=1
\iff e_{M,w}\in\mathcal{B}_{\mathrm{bw}}.
\]
If $B_{\mathrm{bw}}$ were decidable on all total computable pipelines,
then $H$ would be decidable, contradiction.
\end{proof}

\begin{remark}
\label{rem:rice_alternative}
Alternatively, $B_{\mathrm{bw}}$ is a non-trivial semantic property of total computable functions, so undecidability follows from Rice's theorem~\cite{rice1953}.
The explicit reduction additionally shows $B_{\mathrm{bw}}$ is $\Sigma_1^0$-hard.
\end{remark}

\begin{proposition}[Semi-decidability of positive instances]
\label{prop:semi_decidable_positive}
Under the promise that input indices denote total computable pipelines,
the language $\mathcal{B}_{\mathrm{bw}}$ is recursively enumerable.
\end{proposition}

\begin{proof}
Fix an effective enumeration $(x_j)_{j\in\N}$ of $\Str$.
Given an index $e$, run a dovetailing procedure over stage $t=0,1,2,\dots$:
for each $j\le t$, simulate $f_e(x_j)$ for $t$ steps.
Whenever a simulation halts with output $y$, compute $\valscore_{\mathrm{bw}}(x_j,y)$.
If the value is positive, accept.

If $e\in\mathcal{B}_{\mathrm{bw}}$, there exists some witness $x_j$ with
$\valscore_{\mathrm{bw}}(x_j,f_e(x_j))>0$;
since $f_e$ is total, that computation halts in finite time,
so the dovetailer eventually accepts.
If $e\notin\mathcal{B}_{\mathrm{bw}}$, no accepting witness exists, so the machine runs forever.
Hence $\mathcal{B}_{\mathrm{bw}}$ is recursively enumerable.
\end{proof}

\subsection{Decidability on finite domains}

\begin{proposition}[Finite-domain decidability with exponential cost]
\label{prop:finite_domain_decidable}
Fix maximum input length $n$ and define
$\mathcal{D}_n:=\Sig^{\le n}$,
$N_n:=|\mathcal{D}_n|=\sum_{j=0}^n|\Sig|^j$.
Let $T_f(n)$ be the worst-case time to evaluate $f(x)$ for $|x|\le n$,
and $T_{\mathrm{bw}}(n)$ the worst-case time to evaluate
$\valscore_{\mathrm{bw}}(x,y)$ for $|x|\le n$ and $y=f(x)$.
Then $B_{\mathrm{bw}}(f)$ on $\mathcal{D}_n$ is decidable in
\[
O\bigl(N_n\cdot(T_f(n)+T_{\mathrm{bw}}(n))\bigr)
=O\bigl(|\Sig|^n\cdot(T_f(n)+T_{\mathrm{bw}}(n))\bigr).
\]
\end{proposition}

\begin{proof}
Because $\mathcal{D}_n$ is finite, exhaustive search terminates.
Algorithm:
enumerate all $x\in\mathcal{D}_n$;
for each $x$, compute $y=f(x)$ and then
$\delta_x:=\valscore_{\mathrm{bw}}(x,y)$.
Return $1$ iff some $\delta_x>0$.

Correctness follows directly from Definition~\ref{def:bias_predicate_bw} with domain restricted to $\mathcal{D}_n$.
Running time is at most $N_n$ iterations, each costing
$O(T_f(n)+T_{\mathrm{bw}}(n))$.
\end{proof}

\begin{remark}
\label{rem:finite_domain_interpretation}
Theorem~\ref{thm:rice_fixed} and Proposition~\ref{prop:finite_domain_decidable} establish a sharp boundary:
overspecification detection is decidable iff the input domain is finite.
This is structurally different from cell-probe or amortized lower bounds for dynamic data structures~\cite{fredman1989,sleator1985}: those are complexity barriers on finite workloads, while our result is a computability barrier on selector families.
\end{remark}

\section{Main result II: Fixed-point barrier to conservative repair}
\label{sec:correction}

A trivial repair operator can force $\valscore=0$ everywhere by outputting a generic baseline implementation, but this collapses the selection pipeline.
We formalize a \emph{conservative} constraint and show that overspecified fixed points necessarily exist.

\subsection{Conservative program transformers}

\begin{definition}[Computable repair operator]
\label{def:phi_index}
A \emph{repair operator} is a total computable function $\Phi:\N\to\N$
mapping pipeline indices to pipeline indices.
\end{definition}

\begin{definition}[Conservativeness]
\label{def:conservative}
$\Phi$ is \emph{conservative} if
$B_{\mathrm{bw}}(f_e)=0 \Rightarrow \Phi(e)=e$.
\end{definition}

\begin{remark}
\label{rem:conservative_strength}
Conservativeness encodes \enquote{do not modify an already evidence-aligned program}---the program-transformation analogue of minimal repair in consistent query answering~\cite{arenas1999,bertossi2011}.
\end{remark}

\begin{definition}[Uniform elimination guarantee]
\label{def:uniform_elimination}
A repair operator $\Phi$ has \emph{uniform elimination} if
\[
\forall e\in\N,\quad B_{\mathrm{bw}}(f_{\Phi(e)})=0.
\]
\end{definition}

\subsection{Overspecified fixed points via Kleene's recursion theorem}

\begin{lemma}[Self-referential gadget family]
\label{lem:gadget_family}
Fix a total computable repair operator $\Phi$ and witnesses $x_0,s_0,y^+,\epsilon,\#$ from Lemma~\ref{lem:nontrivial_bw}.
Then there exists a total computable map $G:\N\to\N$ such that, for every $e\in\N$,
\[
f_{G(e)}(x)=
\begin{cases}
y^+, & \Phi(e)=e\text{ and }x=x_0\#^n\text{ for some }n\in\N,\\
\epsilon, & \text{otherwise.}
\end{cases}
\]
\end{lemma}

\begin{proof}
Define the computable predicate
\[
P(e,x):=\bigl[\Phi(e)=e\bigr]\wedge\bigl[\exists n\in\N: x=x_0\#^n\bigr].
\]
Because $\Phi$ is total computable, $\Phi(e)$ is computable for each $e$;
membership in $\{x_0\#^n:n\in\N\}$ is decidable by a linear-time scan.
Hence the program
\[
Q(e,x):=\begin{cases}
y^+, & P(e,x)=1,\\
\epsilon, & P(e,x)=0,
\end{cases}
\]
is total computable in both arguments.
By the $s$-$m$-$n$ theorem, there is a total computable index map $G$ with
$f_{G(e)}(x)=Q(e,x)$ for all $e,x$.
\end{proof}

\begin{theorem}[Overspecified fixed point under conservative repair]
\label{thm:biased_fixed_point}
Assume $B_{\mathrm{bw}}$ is non-trivial (Lemma~\ref{lem:nontrivial_bw}).
For any total computable conservative repair operator $\Phi$,
there exists $e^\ast\in\N$ with
\begin{equation}
\label{eq:biased_fp}
\Phi(e^\ast)=e^\ast
\quad\text{and}\quad
B_{\mathrm{bw}}(f_{e^\ast})=1.
\end{equation}
\end{theorem}

\begin{proof}
Fix $\Phi$ and obtain $G$ from Lemma~\ref{lem:gadget_family}.
Apply Kleene's recursion theorem (Appendix~\ref{app:recursion}) to $G$:
there exists $e^\ast$ such that
\[
f_{e^\ast}=f_{G(e^\ast)}.
\]
By Lemma~\ref{lem:gadget_family}, for every $x$,
\[
f_{e^\ast}(x)=
\begin{cases}
y^+, & \Phi(e^\ast)=e^\ast\text{ and }x=x_0\#^n\text{ for some }n,\\
\epsilon, & \text{otherwise.}
\end{cases}
\]

\textbf{Claim 1:} $\Phi(e^\ast)=e^\ast$.
Assume toward contradiction $\Phi(e^\ast)\ne e^\ast$.
Then the display above gives $f_{e^\ast}(x)=\epsilon$ for all $x$,
so $B_{\mathrm{bw}}(f_{e^\ast})=0$.
By conservativeness (Definition~\ref{def:conservative}),
$\Phi(e^\ast)=e^\ast$, contradiction.
Hence Claim~1 holds.

\textbf{Claim 2:} $B_{\mathrm{bw}}(f_{e^\ast})=1$.
By Claim~1 and the gadget definition,
$f_{e^\ast}(x_0\#^n)=y^+$ for all $n\in\N$.
Using Lemma~\ref{lem:nontrivial_bw},
\[
\valscore_{\mathrm{bw}}(x_0\#^n,f_{e^\ast}(x_0\#^n))
=\sum_{s\in\Smodel(x_0\#^n)\setminus\Warrant(x_0\#^n)}\Vrel(y^+,s)
\ge \Vrel(y^+,s_0)=1>0.
\]
Therefore $B_{\mathrm{bw}}(f_{e^\ast})=1$.

Combining Claims~1 and~2 yields~\eqref{eq:biased_fp}.
\end{proof}

\begin{corollary}[No conservative total repair operator can be complete]
\label{cor:no_uniform_conservative}
No total computable repair operator can simultaneously satisfy conservativeness
(Definition~\ref{def:conservative}) and uniform elimination
(Definition~\ref{def:uniform_elimination}).
\end{corollary}

\begin{proof}
Assume $\Phi$ is total computable and conservative.
By Theorem~\ref{thm:biased_fixed_point}, there exists $e^\ast$ with
$\Phi(e^\ast)=e^\ast$ and $B_{\mathrm{bw}}(f_{e^\ast})=1$.
Hence
\[
B_{\mathrm{bw}}(f_{\Phi(e^\ast)})
=B_{\mathrm{bw}}(f_{e^\ast})
=1,
\]
which violates uniform elimination.
\end{proof}

\begin{remark}
\label{rem:where_holds_correction}
Theorem~\ref{thm:biased_fixed_point} is worst-case:
it constructs a specific adversarial selection pipeline evading \emph{any} conservative repair operator.
Practical repair heuristics can still reduce overspecification on typical workloads.
\end{remark}

\begin{remark}
\label{rem:where_fails_correction}
Dropping conservativeness breaks the theorem: $\Phi(e)=e_0$ (a constant baseline selector) has a single unoverspecified fixed point---but destroys all representation sensitivity.
\end{remark}

\subsection{The three-way algorithmic trade-off}
\label{sec:practical_conservativeness}

Combining Theorems~\ref{thm:rice_fixed},~\ref{thm:biased_fixed_point}, and Corollary~\ref{cor:no_uniform_conservative},
any repair algorithm for adaptive representation selection faces:
\begin{enumerate}[label=(\alph*),leftmargin=*]
\item \textbf{Abandon conservativeness.} Modify all pipelines, risking degradation of already evidence-aligned selectors.
\item \textbf{Abandon completeness.} Accept that some overspecified pipelines pass uncorrected.
\item \textbf{Restrict the domain.} Operate on finite instance families (Proposition~\ref{prop:finite_domain_decidable}), accepting exponential cost.
\end{enumerate}

Practical methods for algorithm selection and benchmark-driven tuning effectively choose (b)~\cite{rice1976,mcgeoch2012,kraska2018}.
Our results explain why this is the strongest uniform strategy available without abandoning either conservativeness or domain generality.

This trade-off is structurally distinct from data-structure lower bounds~\cite{fredman1989,holm2001}: those bounds constrain efficiency within decidable settings; our barrier constrains possibility on unbounded selector families.

\subsection{Connection to minimal repair}
\label{sec:cqa_connection}

Conservativeness (Definition~\ref{def:conservative}) is the selector analogue of \emph{minimality} in repair-based reasoning~\cite{arenas1999,bertossi2011}.
Theorem~\ref{thm:biased_fixed_point} says: under this discipline, selectors exist that are their own \enquote{repair} yet remain structurally overspecified.

\section{Conclusion}
\label{sec:conclusion}

Working under a model in which structural preferences propagate through both benchmark aggregation and pairwise score fitting (Propositions~\ref{prop:social_inheritance_profile}--\ref{prop:btl_inheritance}), we identified two algorithmic barriers to repair in adaptive data-structure selection:
\begin{enumerate}[label=(\roman*),leftmargin=*]
\item \textbf{Decidability boundary.} Overspecification detection is undecidable on unbounded domains but decidable at exponential cost on finite domains.
\item \textbf{Fixed-point barrier.} Conservative repair operators admit overspecified fixed points via Kleene's recursion theorem.
\end{enumerate}

These barriers are qualitatively different from classical lower bounds in data structures and algorithm analysis~\cite{fredman1989,holm2001}.
Those bounds constrain the \emph{efficiency} of supporting operations.
Our results constrain the \emph{possibility} of uniformly detecting and repairing unwarranted structural commitment across selector families.
The resulting three-way trade-off between conservativeness, completeness, and domain restriction provides a principled framework for evaluating any adaptive representation-selection pipeline in graph algorithms, string algorithms, and dynamic data structures.


\appendix
\setcounter{section}{0}
\renewcommand{\thesection}{\Alph{section}}

\section{Majority Linkage}
\label{app:majority}

\begin{proposition}[Strict majority implies decisive power]
\label{prop:majority_implies_decisive}
Under pairwise majority aggregation with odd $k$, any coalition $C\subseteq[k]$ with $|C|>\frac{k}{2}$ is decisive for every ordered pair $(y,y')$.
If every evaluator in $C$ is signature-monotone on $(y,y')$, Proposition~\ref{prop:social_inheritance_profile} yields inheritance.
\end{proposition}

\begin{proof}
Fix an instance $x$ and candidate set $Y_x$ with $|Y_x|\ge 3$.
Let $A$ be the pairwise majority rule.
Take any coalition $C\subseteq[k]$ with $|C|>\frac{k}{2}$.
Suppose every evaluator in $C$ ranks $y\succ_i y'$.
Then strictly more than half of all evaluators rank $y$ above $y'$.
By majority aggregation, the aggregate outcome satisfies $y\succ_A y'$ regardless of the rankings of evaluators outside $C$.
Hence $C$ is decisive for $(y,y')$.

If, in addition, every evaluator in $C$ is signature-monotone on $(y,y')$, then
$\valscore(x,y)>\valscore(x,y')$ implies $y\succ_i y'$ for all $i\in C$,
and decisiveness gives $y\succ_A y'$.
This is exactly Proposition~\ref{prop:social_inheritance_profile}.
\end{proof}

\section{Pairwise-Fit Consistency}
\label{app:btl_consistency}

\begin{proposition}[Well-specified Bradley--Terry consistency]
\label{prop:btl_well_spec}
Under correct specification and positive pair coverage,
the Bradley--Terry estimator recovers true score differences for all sufficiently large samples~\cite{bradley1952,hunter2004}.
\end{proposition}

\begin{proposition}[Misspecified pairwise fitting]
\label{prop:btl_misspec}
Under misspecification, consistent $M$-estimators converge to the $\KL$ projection of the true pairwise law onto the fitted pairwise-comparison family.
\end{proposition}

\section{Kleene Fixed-Point Theorem}
\label{app:recursion}

\begin{lemma}[Self-application operator from $s$-$m$-$n$]
\label{lem:self_application_operator}
There exists a total computable function $q:\N\to\N$ such that for all $n,x\in\N$,
\[
\varphi_{q(n)}(x)\simeq\varphi_{\varphi_n(n)}(x),
\]
where $\simeq$ denotes equality of partial computable functions
(both sides undefined or both defined with equal value).
\end{lemma}

\begin{proof}
Define a partial computable binary function $H:\N\times\N\rightharpoonup\N$ by
\[
H(n,x):=\varphi_{\varphi_n(n)}(x).
\]
Operationally: simulate $\varphi_n(n)$; if it halts with output $m$,
then simulate $\varphi_m(x)$ and return its output.
If either simulation diverges, $H(n,x)$ is undefined.
Thus $H$ is partial computable.

By the $s$-$m$-$n$ theorem, there exists a total computable function $q$ such that
\[
\forall n,x,\quad \varphi_{q(n)}(x)\simeq H(n,x)
=\varphi_{\varphi_n(n)}(x).
\]
This is exactly the stated property.
\end{proof}

\begin{theorem}[Kleene recursion theorem]
\label{thm:kleene_formal}
For every total computable $G:\N\to\N$,
there exists $e^\ast\in\N$ with $\varphi_{e^\ast}=\varphi_{G(e^\ast)}$.
\end{theorem}

\begin{proof}
Fix total computable $G:\N\to\N$ and let
$q$ be the total computable map from Lemma~\ref{lem:self_application_operator}.

Define
\[
h(n):=G(q(n)).
\]
Since both $G$ and $q$ are total computable, $h$ is total computable.
By effective enumeration of partial computable functions,
there exists an index $p\in\N$ such that
\[
\forall n\in\N,\quad \varphi_p(n)=h(n)=G(q(n)).
\]

Now set
\[
e^\ast:=q(p).
\]
We verify that $\varphi_{e^\ast}=\varphi_{G(e^\ast)}$.
For arbitrary $x\in\N$,
\begin{align*}
\varphi_{e^\ast}(x)
&=\varphi_{q(p)}(x) \\
&\simeq \varphi_{\varphi_p(p)}(x)
&&\text{(Lemma~\ref{lem:self_application_operator})} \\
&=\varphi_{h(p)}(x)
&&\text{(definition of $p$)} \\
&=\varphi_{G(q(p))}(x)
&&\text{(definition of $h$)} \\
&=\varphi_{G(e^\ast)}(x).
\end{align*}
Because $x$ was arbitrary, $\varphi_{e^\ast}=\varphi_{G(e^\ast)}$.
\end{proof}

\section{Finite-Domain Detection Algorithm}
\label{app:finite_domain_algo}

\begin{algorithm}[H]
\caption{Finite-domain overspecification detection}
\label{alg:finite_bias}
\begin{algorithmic}[1]
\Procedure{DecideOverspecification}{$f,n$}
  \ForAll{$x\in\Sig^{\le n}$}
    \If{$\valscore_{\mathrm{bw}}(x,f(x)) > 0$}
      \Return $1$
    \EndIf
  \EndFor
  \Return $0$
\EndProcedure
\end{algorithmic}
\end{algorithm}


\begin{thebibliography}{99}

\bibitem{cormen2009}
Thomas~H. Cormen, Charles~E. Leiserson, Ronald~L. Rivest, and Clifford Stein.
\newblock \emph{Introduction to Algorithms}.
\newblock MIT Press, third edition, 2009.

\bibitem{rice1976}
John~R. Rice.
\newblock The algorithm selection problem.
\newblock \emph{Advances in Computers}, 15:65--118, 1976.

\bibitem{sleator1983}
Daniel~D. Sleator and Robert~E. Tarjan.
\newblock A data structure for dynamic trees.
\newblock \emph{Journal of Computer and System Sciences}, 26(3):362--391, 1983.

\bibitem{sleator1985}
Daniel~D. Sleator and Robert~E. Tarjan.
\newblock Self-adjusting binary search trees.
\newblock \emph{Journal of the ACM}, 32(3):652--686, 1985.

\bibitem{holm2001}
Jacob Holm, Kristian de~Lichtenberg, and Mikkel Thorup.
\newblock Poly-logarithmic deterministic fully-dynamic algorithms for connectivity, minimum spanning tree, 2-edge, and biconnectivity.
\newblock \emph{Journal of the ACM}, 48(4):723--760, 2001.

\bibitem{demetrescu2004}
Camil Demetrescu and Giuseppe~F. Italiano.
\newblock A new approach to dynamic all pairs shortest paths.
\newblock \emph{Journal of the ACM}, 51(6):968--992, 2004.

\bibitem{kmp1977}
Donald~E. Knuth, James~H. Morris, and Vaughan~R. Pratt.
\newblock Fast pattern matching in strings.
\newblock \emph{SIAM Journal on Computing}, 6(2):323--350, 1977.

\bibitem{manber1993}
Udi Manber and Gene Myers.
\newblock Suffix arrays: A new method for on-line string searches.
\newblock \emph{SIAM Journal on Computing}, 22(5):935--948, 1993.

\bibitem{ukkonen1995}
Esko Ukkonen.
\newblock On-line construction of suffix trees.
\newblock \emph{Algorithmica}, 14(3):249--260, 1995.

\bibitem{mcgeoch2012}
Catherine~C. McGeoch.
\newblock \emph{A Guide to Experimental Algorithmics}.
\newblock Cambridge University Press, 2012.

\bibitem{kraska2018}
Tim Kraska, Alex Beutel, Ed~H. Chi, Jeffrey Dean, and Neoklis Polyzotis.
\newblock The case for learned index structures.
\newblock In \emph{Proceedings of SIGMOD}, pages 489--504, 2018.

\bibitem{fredman1989}
Michael~L. Fredman and Michael~E. Saks.
\newblock The cell probe complexity of dynamic data structures.
\newblock In \emph{Proceedings of STOC}, pages 345--354, 1989.

\bibitem{bradley1952}
Ralph~A. Bradley and Milton~E. Terry.
\newblock Rank analysis of incomplete block designs: {I}. The method of paired comparisons.
\newblock \emph{Biometrika}, 39(3--4):324--345, 1952.

\bibitem{hunter2004}
David~R. Hunter.
\newblock MM algorithms for generalized Bradley--Terry models.
\newblock \emph{The Annals of Statistics}, 32(1):384--406, 2004.

\bibitem{arenas1999}
Marcelo Arenas, Leopoldo~E. Bertossi, and Jan Chomicki.
\newblock Consistent query answers in inconsistent databases.
\newblock In \emph{Proceedings of PODS}, pages 68--79, 1999.

\bibitem{bertossi2011}
Leopoldo~E. Bertossi.
\newblock \emph{Database Repairing and Consistent Query Answering}.
\newblock Morgan \& Claypool, 2011.

\bibitem{rice1953}
Henry~G. Rice.
\newblock Classes of recursively enumerable sets and their decision problems.
\newblock \emph{Transactions of the American Mathematical Society}, 74(2):358--366, 1953.

\bibitem{kahneman1979}
Daniel Kahneman and Amos Tversky.
\newblock Prospect theory: An analysis of decision under risk.
\newblock \emph{Econometrica}, 47(2):263--291, 1979.

\bibitem{tversky1992}
Amos Tversky and Daniel Kahneman.
\newblock Advances in prospect theory: Cumulative representation of uncertainty.
\newblock \emph{Journal of Risk and Uncertainty}, 5(4):297--323, 1992.

\end{thebibliography}
\end{document}